\documentclass[a4paper,12pt]{article}
\pdfoutput=1

\usepackage{amsthm}
\usepackage{amsmath}
\usepackage{epsfig}
\usepackage{amssymb}
\usepackage{latexsym}
\usepackage[T1]{fontenc}
\usepackage[utf8]{inputenc}
\usepackage{color}

\author{
Barbora Candráková,
Robert Lukoťka
\\[3mm]
\\{\tt barbora.candrakova@gmail.com}
\\{\tt lukotka@dcs.fmph.uniba.sk}
\\[5mm]
Faculty of Mathematics, Physics and Informatics\\
Comenius University \\
Mlynská dolina, 842 48 Bratislava\\
}
\title{Cubic TSP -- a $1.3$-approximation}

\theoremstyle{definition}
\newtheorem{definition}{Definition}
\newtheorem{theorem}[definition]{Theorem}
\newtheorem{observation}[definition]{Observation}
\newtheorem{corollary}[definition]{Corollary}
\newtheorem{lemma}[definition]{Lemma}

\newcommand{\V}{V \hspace{-2pt}}
\newcommand{\SV}{S \hspace{-2pt}}

\let\epsilon=\varepsilon

\let\phi = \varphi

\begin{document}

\maketitle

\abstract{  
We prove that every simple bridgeless cubic graph with $n \ge 8$ vertices has a travelling salesman tour of length at most $1.3\cdot n - 2$,
which can be constructed in polynomial time.
\\


}

\section{Introduction}
\emph{Travelling salesman problem} (\emph{TSP}) is one of the most studied topics in both computer science and combinatorial optimization. Given a set of points and the distance 
between every pair of them, the goal is to find the shortest route that visits each point exactly once and ends in the starting point. 
A natural restriction on TSP is to require that the distances between points satisfy the triangle inequality, which is called \emph{metric TSP}.
Metric TSP is known to be NP-hard as well as
NP-hard to approximate with a ratio better than $185/184$ \cite{L}. 
Christofides \cite{C} showed that it is possible to approximate metric TSP with a ratio $1.5$ in polynomial time.

Many other restrictions of the problem, such as Euclidean TSP, $(1,2)$-TSP, or graphic TSP
have been extensively investigated \cite{AS, BK, MS}.
The last restriction, \emph{graphic TSP},
asks to find a closed tour as short as possible, such that it contains all vertices of an unweighted graph. (If we allow weights in the graph, then the problem is equivalent to metric TSP.) We call such tour a \emph{travelling salesman tour} or a \emph{TSP tour} for short. 
Seb\" o and Vygen found a $1.4$-approximation
algorithm for the graphic TSP. For more information on the approximation algorithms for TSP we refer to the survey of Vygen \cite{V}.

This paper is devoted to instances of graphic TSP, where the graph is 
simple cubic and bridgeless. 
Since each bridge has to be used twice in any closed trail, 
TSP on cubic graphs with bridges naturally reduces to TSP on subcubic graphs.
M\" omke and Svennson \cite{MS} proved, that 
a connected subcubic graph on $n$ vertices has TSP tour of length at most 
$4/3 \cdot n-2/3$, which is a tight bound. 
Even if the graph is cubic and has multiedges, then we can easily modify the construction of graphs attaining the bound $4/3 \cdot n-2/3$ from \cite{BSSS}
to use multiedges since we can use them to simulate two consecutive vertices of degree $2$.
Therefore, to make further improvement one has to consider simple bridgeless
cubic graphs.

Boyd et al. \cite{BSSS} proved, that a simple connected bridgeless cubic graph 
on $n$ vertices has a TSP tour of length at most $4/3 \cdot n -2$, provided $n \ge 6$. 
Correa, Larré, and Soto \cite{CLS} improved the result to $(4/3 - 1/8754) \cdot n$.
Zuylen \cite{zuylen} improved this further to $(4/3 - 1/61236) \cdot n$.
If we restrict ourselves to simple bipartite cubic graph, then we get TSP tour of length at most $9/7 \cdot n$ as shown by Karp and Ravi \cite{KR}.
Our main result is the following.
\begin{theorem}\label{thm1}
Let $G$ be a connected bridgeless cubic graph on $n$ vertices, where $n \ge 8$.
Then $G$ has a travelling salesman tour of length at most $1.3 \cdot n-2$.
\end{theorem}
\noindent The proof of the theorem is constructive and the tour can be constructed 
in polynomial time. 
The algorithm provides the best polynomial time approximation of TSP on simple cubic graphs known at present.
Very recently Dvorak, Kral, and Mohar improved the result by showing that 
every bridgeless cubic graph has a TSP tour of length at most $9/7\cdot n-1$ \cite{DKM}.

The paper is organized as follows. In section~\ref{sec2} we define cost of an even factor and restate 
Theorem~\ref{thm1} in this new setting as Theorem~\ref{thm2}.
We show that none of the smallest counterexamples to Theorem~\ref{thm2}  
contains several defined subgraphs.
In section~\ref{sec3} we define special kinds of even factors, called bounded even factors, 
where the position of isolated vertices is restricted and the isolated vertices are assigned to the circuits 
of the even factor.
We restate Theorem~\ref{thm2} as Theorem~\ref{thm} by using bounded even factors  
(which are even factors satisfying an additional property) instead of even factors 
and by considering only graphs that contain no reducible subgraphs.
Moreover, we define operations on bounded even factors, called swaps, which decrease the cost of the factor.
In section~\ref{sec4} we prove our key lemma (Lemma~\ref{factorlemma}) that allows us to find a 
suitable $2$-factor $F$, which will be our starting bounded even factor. 
The lemma is stated generally to allow its future reuse 
(e.g. graphs can have reducible configurations).
Section~\ref{sec5} describes how the swaps are carried out on $F$ and how the cost is distributed among 
the vertices of the graph. We bound the cost of the vertices according to the requirements of 
Lemma~\ref{factorlemma}. In section~\ref{sec6} we finish the proof of Theorem~\ref{thm} and thus also
Theorem~\ref{thm2} and Theorem~\ref{thm1}. We discuss how our proof can be converted 
into a polynomial algorithm to find the TSP tour.

\section{Even factors, reductions}\label{sec2}

First we define several basic standard graph theory notions used in the paper.
All graphs in the rest of the paper are simple unless explicitly stated otherwise (we keep using the word "simple" occasionally to emphasise it).
Let $G$ be a graph. We denote the vertex set of $G$ by $V(G)$ and edge set of $G$ by $E(G)$. 
A \emph{circuit} is a connected graph such that each vertex has degree $2$.
A \emph{$k$-circuit} is a circuit that contains exactly $k$ vertices.
A \emph{boundary} of a vertex set $W \subseteq V(G)$ (or of a subgraph $H$ of $G$), denoted by $\partial(W)$ ($\partial(H)$), 
is the set of edges with precisely one end-vertex in the $W$ (in $V(H)$). 
A set of edges is \emph{independent} if no two edges are incident to a common vertex
(thus an \emph{independent boundary} is a boundary that contains independent edges).
Let $H$ be a subgraph of $G$. A \emph{contraction} of $H$ in $G$, denoted by $G/H$, is a graph created
by contracting all the edges in $H$.
A \emph{spanning subgraph} of $G$ is a subgraph of $G$ that contains all its vertices. 
A \emph{perfect matching} of $G$ is a spanning subgraph of $G$ such that
each vertex has degree $1$.
A \emph{$2$-factor} of $G$ is a spanning subgraph of $G$ such that
each vertex has degree $2$.
An \emph{Eulerian graph} is a connected graph such that the degree of each vertex 
is even (an isolated vertex is an Eulerian graph). 
An \emph{even factor} of $G$ is a spanning subgraph $F$ of $G$ such that
each component of $F$ is Eulerian.
If $G$ is cubic, then each component of $F$ is either a circuit or an isolated vertex.
For an even factor $F$ we denote the set of isolated vertices of $F$ by $V_F$ and 
the set of circuits of $F$ by ${\cal C}_F$.

\medskip

Let $G$ be a bridgeless cubic graph on $n$ vertices and let $F$ be an even factor of $G$. We can naturally construct TSP tour of $G$ using $F$. First we find a spanning tree $T$ in the graph $G/F$. Then we take the Eulerian multigraph that contains all edges of $F$ once and all edges of $T$ twice. The Eulerian tour of this graph yields a TSP tour in $G$. Since $T$ has $|V_F|+|{\cal C}_F| - 1$ edges
and $F$ has $n-|V_F|$ edges, the total length of the TSP tour is $n+2\cdot |{\cal C}_F|+|V_F| - 2$.

We can express the length of such TSP tour as the \emph{cost} of the even factor $F$ as follows. 
For each component $H$ of $F$ we define the cost of $H$ as $c(H)=|E(H)|+2$,
and the cost of $F$ as $c(F)=\sum_{H \in V_F \cup {\cal C}_F} c(H)$.
Clearly, $c(F)=n+2\cdot |{\cal C}_F|+ |V_F|$, therefore, $G$ has a TSP tour of length at most $c(F) -2$.
Thus we can reformulate Theorem~\ref{thm1} as follows.
\begin{theorem}
Let $G$ be a simple connected bridgeless cubic graph on $n$ vertices, where $n \ge 8$.
Then $G$ has an even factor with cost at most $1.3 \cdot n$.
\label{thm2}
\end{theorem}

\medskip

We start by proving that several subgraphs cannot be present in 
any smallest counterexample to Theorem~\ref{thm2}.
An \emph{8-diamond} is an $8$-circuit with $3$ chords.
A \emph{6-diamond} is a $6$-circuit with $2$ chords such that the circuit is not contained in any 8-diamond.
A \emph{4-diamond} is a $4$-circuit with a chord such that the circuit is not contained in any 6-diamond.
The following induced subgraphs will be called \emph{reducible subgraphs}:
\begin{itemize}
\item \emph{type 1 reducible subgraph}: a $5$-circuit with a chord and an independent boundary (Figure~\ref{fig1} left); 
\item \emph{type 2 reducible subgraph}: an 8-diamond;
\item \emph{type 3 reducible subgraph}: a $7$-circuit with $2$ chords forming a $4$-diamond and with an independent boundary (Figure~\ref{fig1} right).
\item \emph{type 4 reducible subgraph}: a $6$-circuit with exactly one chord that connects two vertices of the $6$-circuit in distance $2$.
\end{itemize}
A bridgeless cubic graph is called \emph{irreducible} if it contains no reducible subgraph.
\begin{figure}
\center
\includegraphics[scale=1.6]{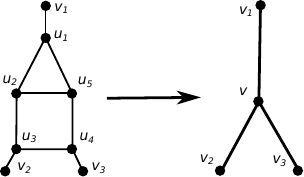} \ \ \ \ \ \ \
\includegraphics[scale=1.6]{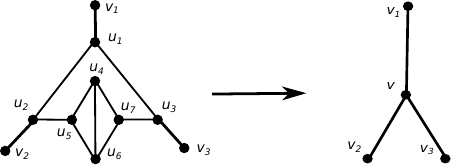} \\ \bigskip
\includegraphics[scale=1.6]{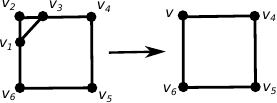}
\caption{Reducible subgraphs of type 1, 3, and 4 with their reductions.}
\label{fig1}
\end{figure} 
We show that in order to prove Theorem~\ref{thm2} it is sufficient to work with irreducible graphs.
The following observation can be easily checked by hand.
\begin{observation}
Let $G$ be a connected bridgeless cubic multigraph on at most $8$ vertices, 
 with no more than one pair of multiedges, and let $e\in E(G)$. 
Then $G$ contains a Hamiltonian circuit containing $e$.
\label{obs1}
\end{observation}

\begin{lemma}
The smallest counterexample to Theorem~\ref{thm2} does not contain any reducible subgraph. 
\label{irred}
\end{lemma}
\begin{proof}
Let $G$ be the smallest counterexample to Theorem~\ref{thm2}. Suppose for contradiction that $G$ contains a reducible subgraph. We show that regardless of which type the subgraph is, a contradiction results.

Suppose that $G$ contains a type 1 reducible subgraph $S$. 
Let as denote the vertices of $S$ by $u_1$,$u_2$,$u_3$,$u_4$, and $u_5$ with the chord $u_2u_5$ and the vertices adjacent to $S$ by $v_1$,$v_2$,$v_3$,$v_4$, and $v_5$ respectively as shown in Figure~\ref{fig1}, left.
Let $G'$ be a graph obtained from $G$ by contracting $S$ into one vertex $v$ (Figure~\ref{fig1}, left). 
If $G'$ has at most $8$ vertices, then we easily extend a Hamiltonian circuit of $G'$, whose existence is guaranteed by Observation~\ref{obs1}, into a Hamiltonian circuit of $G$. 
Otherwise, let $F'$ be an even factor of $G'$ satisfying Theorem~\ref{thm2}.
We create an even factor $F$ of $G$ from $F'$ as follows.

If $v$ is contained in a circuit of $F'$, then we may without loss of generality assume that
either $F'$ contains $vv_1$ and $vv_2$ or $F'$ contains $vv_2$ and $vv_3$.
If $F'$ contains $vv_1$ and $vv_2$, then we define $F=(F'- v_1vv_2) \cup v_1u_1u_2u_5u_4u_3v_2$.
If $F'$ contains $vv_2$ and $vv_3$, then we define $F=(F'- v_2vv_3) \cup v_2u_3u_2u_1u_5u_4v_3$.
In both cases $|V(G)| - |V(G')| = 4$ and $c(F) - c(F') = 4$. Since $F'$ satisfies $c(F')\le 1.3\cdot|V(G')|$, we get $c(F)\le 1.3\cdot|V(G)| - 1.2$, which contradicts the fact that $G$ is a counterexample to Theorem~\ref{thm2}.
If $v$ is isolated in $F'$, then we set $F=F' \cup u_1u_2u_3u_4u_5u_1$.
Now $|V(G)| - |V(G')| = 4$ and $c(F) - c(F') = 5$.
Since $F'$ satisfies $c(F')\le 1.3\cdot|V(G')|$, we get 
$c(F)\le 1.3\cdot|V(G)| - 0.2$, which contradicts the fact that $G$ is a counterexample to Theorem~\ref{thm2}.

Suppose that $G$ contains a type 2 reducible subgraph $S^0$. 
Let us denote the vertices of degree $2$ of $S_0$ by $v_1$ and $v_2$. 
We denote the vertices adjacent to $S^0$ so that the $\partial(S^0)$ contains edges $v_1v_1^0$ and
$v_2v_2^0$. Let $C^0$ denote the $8$-circuit of $S_0$.
Let $P^0$ denote a Hamiltonian path connecting $v_1$ and $v_2$. The existence of such path is implied by Observation~\ref{obs1}. Indeed, we add a new edge $e$ between $v_1$ and $v_2$ to $S^0$ 
(this may produce a multighraph) and then we require $e$ to be in a Hamiltonian circuit. 
Removing $e$ from the Hamiltonian circuit creates a Hamiltonian path.

For $i>0$,
if $v^{i-1}_1$, $v^{i-1}_2$ are adjacent, then we define $S^i$, $v^i_1$, $v^i_2$, $P^i$, and $C^i$  as follows. 
Let $S^i$ be the subgraph of $G$ induced by $V(S^{i-1}) \cup \{v_1^{i-1}, v_2^{i-1}\}$,
let $P_i=v^{i-1}_1P^{i-1}v^{i-1}_2$, let $C^i=v^{i-1}_1P^{i-1}v^{i-1}_2v^{i-1}_1$, 
and let $v^i_1$ and $v^i_2$ be the vertices outside $S^i$ adjacent to $v^{i-1}_1$ and $v^{i-1}_2$, respectively
(they exist because $G$ is simple). As $G$ is bridgeless $v^i_1 \neq v^i_2$.
Let $k$ be the highest integer for which $S^k$ is defined.

Let $G'$ be a graph obtained from $G$ by deleting all vertices of $S^k$ and adding the edge $v_1^kv_2^k$. If $G'$ has at most $8$ vertices, 
then using path $P^k$ we can easily extend the Hamiltonian circuit of $G'$ that contains $v_1^kv_2^k$
into Hamiltonian circuit of $G$ which contradicts that $G$ is a counterexample to Theorem~\ref{thm2}. 
Therefore $G'$ has more than $8$ vertices.
Let $F'$ be an even factor of $G'$ that satisfies Theorem~\ref{thm2}.
We create an even factor $F$ of $G$ from $F'$ as follows.
If $v_1^kv_2^k\in F'$, then $F=(F'-v_1^kv_2^k) \cup v_1^kP^kv_2^k$.
We have $|V(G)| - |V(G')| = |V(S^k)|$ and $c(F) - c(F') = |V(S^k)|$. The even factor $F'$ satisfies $c(F')\le 1.3\cdot |V(G')|$, hence we get  $c(F)\le 1.3\cdot |V(G)|-0.3 \cdot |V(S^k)|$, which contradicts the fact that $G$ is the smallest counterexample to Theorem~\ref{thm2}.
If $v_1^kv_2^k \not \in F'$, then $F=F'\cup C^k$.
We have $|V(G)| - |V(G')| = |V(S^k)|$ and $c(F) - c(F') = |V(S^k)|+2$. The even factor $F'$ satisfies $c(F')\le 1.3\cdot |V(G')|$, hence we get  $c(F)\le 1.3\cdot |V(G)|-0.3\cdot |V(S^k)|+2$, which, as $|V(S^k)|\ge 8$, contradicts the fact that $G$ is a counterexample to Theorem~\ref{thm2}.

Suppose that $G$ contains a type 3 reducible subgraph $S$. 
Let us denote the vertices of $S$ by $u_1$,$u_2$,$u_5$,$u_4$,$u_6$,$u_7$, and $u_3$ in successive order so that the two chords forming a $4$-diamond are $u_5u_6$ and $u_4u_7$, and let the vertices adjacent to $S$ be $v_1$,$v_2$, and $v_3$ as in Figure~\ref{fig1}, right.
Let $G'$ be a graph obtained from $G$ by contracting $S$ into one vertex $v$  (Figure~\ref{fig1}, right). 
If $G'$ has at most $8$ vertices, then we easily extend a Hamiltonian circuit of $G'$ into a Hamiltonian circuit of $G$. Otherwise, let $F'$ be an even factor of $G'$ satisfying Theorem~\ref{thm2}.

We may without loss of generality assume that either $v$ is  
isolated vertex of $F'$, or $vv_2,vv_3\in F'$, or $vv_1,vv_2\in F'$. 
If $v$ is isolated in $F'$, then we set $F=F'\cup u_1u_2u_5u_4u_6u_7u_3u_1$. 
If $vv_2,vv_3\in F'$, then we set $F=(F'-v_2vv_3)\cup v_2u_2u_5u_4u_6u_7u_3v_3$ and $u$ is an isolated vertex of $F$. 
If $vv_1,vv_2\in F'$, then we set $F=(F'-v_1vv_2)\cup v_1u_1u_3u_7u_4u_6u_5u_2v_2$. 
We get that $|V(G)| - |V(G')| = 6$ and $c(F) - c(F') \le 7$. Since $F'$ satisfies $c(F')\le 1.3\cdot |V(G')|$ we get a contradiction with $G$ being the smallest counterexample to Theorem~\ref{thm2}.

Finally, suppose that $G$ contains a type 4 reducible subgraph $S$. 
Let us denote the vertices of the $6$-circuit of $S$ by $v_1$, $v_2$, $v_3$, $v_4$, $v_5$, $v_6$
in successive order so that the chord is between $v_1$ and $v_3$ (Figure~\ref{fig1}, bottom).
Let $G'$ be a graph obtained by contracting the triangle $v_1v_2v_3$ into one vertex $v$. 
If $G'$ has at most $8$ vertices, then we easily extend a Hamiltonian circuit of $G'$ into a Hamiltonian circuit of $G$. Otherwise, let $F'$ be an even factor of $G'$ satisfying Theorem~\ref{thm2}. If $v$ is not an isolated vertex of $F'$, then we extend the circuit of $F'$ that passes $v$ by two edges of the triangle $v_1v_2v_3$
and denote the resulting even factor of $G$ by $F$. We have $c(F)=c(F')+2$ and
$|V(G)|=|V(G')|+2$ which contradicts with $G$ being the smallest counterexample to 
Theorem~\ref{thm2}. If $v$ is an isolated vertex of $F'$, then we have up to symmetry three cases to consider: among the edges of the circuit $vv_4v_5v_6$,
either no edge is in $F'$, or only $v_4v_5$ is in $F'$, or $v_4v_5$ and $v_5v_6$
are in $F'$.  If the first case we set $F=F'\cup v_1v_2v_3v_4v_5v_6v_1$, 
in the second case 
$F=(F'-v_4v_5)\cup \{v_4v_3v_2v_1v_6v_5\}$, and in the third case
$F=(F'-v_4v_5v_6)\cup \{v_4v_3v_2v_1v_6\}$ and $v_5$ is an isolated vertex.
In each case $c(F)\leq(F')+2$. Since $F'$ satisfies $c(F')\le 1.3\cdot |V(G')|$ we get a contradiction with $G$ being the smallest counterexample to Theorem~\ref{thm2}.
\end{proof}

\section{BE factors and swaps}\label{sec3}

The cost of the even factor $F$ depends on the presence of isolated vertices and 
the length of circuits from $F$. Short circuits and isolated vertices 
have high cost compared to the number of vertices they cover.
Although isolated vertices may be necessary to produce a cheap even factor (e.g. in the Petersen graph an even factor containing one $9$-circuit and an isolated vertex is cheaper than a $2$-factor which must contain 
two $5$-circuits), 
they have to be used sparsely. To facilitate further analysis, we restrict the presence 
of isolated vertices in even factor and attach them to circuits to distribute the cost.

Let $F$ be an even factor of a cubic graph $G$. Let $v\in V_F$ and let $C\in{\cal C}_F$. 
We say that $v$ is \emph{bounded} to $C$ if
\begin{itemize}
\item $v$ has at least two neighbours in $C$ or
\item $v$ has at least one neighbour in $C$ and one neighbour is an isolated vertex in $F$ bounded to $C$ or
\item $v$ has two neighbours that are isolated vertices in $F$ bounded to $C$.
\end{itemize}
If each isolated vertex of $F$ is bounded to some circuit of $F$,
then we say that $F$ is a \emph{bounded even factor} (\emph{BE factor}) of $G$. 
Since the graph is cubic, each isolated vertex in $F$
can be bounded to at most one circuit of $F$. 
Hence for a BE factor $F$ of $G$, each $v \in V_F$ is bounded to exactly one circuit $C \in {\cal C}_F$.
An \emph{extended circuit} (\emph{X-circuit}) of $F$ is a union of a circuit $C \in {\cal C}_F$ 
with all isolated vertices of $F$ bounded to $C$.
Thus we can alternatively define a BE factor of $G$ as a spanning subgraph of $G$ such that each 
component is an X-circuit. 
For a BE factor $F$ of $G$ let ${\cal X}_F$ be the set of X-circuits of  $F$.
For an X-circuit $X$ let $C_X$ be the circuit of $X$ and 
$V_X$ be the set of isolated vertices of $X$

For each X-circuit $X$ we define the \emph{cost} of $X$ to be $c(X)=|V(X)|+|V_X|+2$,
which corresponds to the definition of the cost in an even factor. 
We can calculate the cost of a BE factor $F$ using any of the following formulas. 
$$
c(F)=\sum_{H \in {\cal{C}}_F \cup V_F} c(H) = \sum_{X \in {\cal{X}_F}} c(X).
$$

We can prove Theorem~\ref{thm2} by proving the following statement taking into account Lemma~\ref{irred}.
\begin{theorem}
Let $G$ be an irreducible simple connected bridgeless cubic graph on $n$ vertices, where $n \ge 8$.
Then $G$ has a BE factor with cost at most $1.3 \cdot n$.
\label{thm}
\end{theorem}

Let $F$ be an even factor of a bridgeless cubic graph.
We can reduce the cost of an even factor by certain swapping operations that merge two or three
circuits of the even factor. 
Beside the cost reduction we prove that each swapping operation preserves boundedness
of the even factor and always merges the X-circuits involved.

Let $C_4=v_1v_2v_3v_4$ be a $4$-circuit of $G$ without a chord.
Assume that all edges from $\partial(C_4)$ belong to two distinct circuits of $F$.
Then exactly two edges of $C_4$ belong to $F$, each one to a distinct circuit of $F$, 
say $v_1v_2$ belongs to a circuit $C_1\in {\cal C}_F$ 
and $v_3v_4$ belongs to a circuit $C_2\in {\cal C}_F$.
We replace the edges $v_1v_2$ and $v_3v_4$ in $F$ with the edges $v_1v_4$ and $v_2v_3$ (Figure \ref{fig2}) 
and denote the resulting even factor by $F'$.
This operation merges the circuits $C_1$ and $C_2$ into one circuit $C$ of $F'$
and leaves the remaining circuits unaffected, thus $c(F)-c(F')=2$. 
We call this operation a \emph{$4$-swap} on $C_4$ and we say that $C_1$ and $C_2$ \emph{participate in the swap}. 
The vertex set $V_F' = V_F$. 
If $v$ is bounded to a circuit $C_v\in {\cal C}_F - C_1 -C_2$ in $F$,
then $v$ is also bounded to $C_v$ in $F'$ and 
if $v$ is bounded to a circuit $C_1$ or $C_2$ in $F$,
then $v$ is bounded to $C$ in $F'$.
Therefore, if $F$ is BE factor, then so is $F'$ and the X-circuits of $F$ are identical 
to X-circuits of $F'$ except for the two X-circuits containing $C_1$ and $C_2$, which are merged into one.
\begin{figure}[htp]
\center
\includegraphics[scale=1]{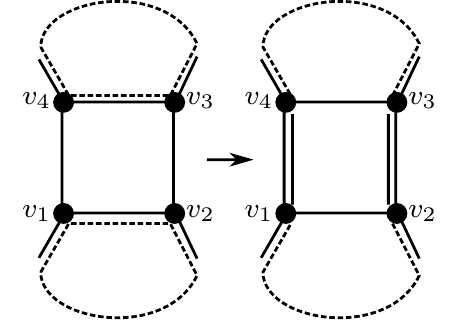}
\caption{A $4$-swap.}
\label{fig2}
\end{figure}

Let $C_5=v_1v_2v_3v_4v_5$ be a $5$-circuit without a chord.
Assume that four edges from $\partial(C_5)$ are in $F$, say that these are the edges incident with $v_1$, $v_2$, $v_3$, and $v_4$, and they belong to two distinct circuits of $F$, $C_1$ and $C_2$. 
Assume first that $v_2v_3 \in F$. Without loss of generality let $v_2v_3\in C_1$. 
Thus $v_1v_5,v_5v_4 \in C_2$.
We define $F'=(F-E(C_5)) \cup v_1v_2 \cup v_3v_4$ (Figure~\ref{fig5swap}, left).
This operation merges the circuits $C_1$ and $C_2$ into one circuit $C$, creates an isolated vertex $v_5$ in $F'$ bounded to $C$, 
and leaves the remaining circuits unaffected, thus $c(F)-c(F') = 1$. 
We call this operation a \emph{$5$-swap of type 1}  on $C_5$ and we say that $C_1$ and $C_2$ \emph{participate in the swap}.
The vertex set  $V_F' = V_F \cup \{v_5\}$. 
If $v$ is bounded to a circuit $C_v\in {\cal C}_F - C_1 -C_2$ in $F$,
then $v$ is bounded to $C_v$ in $F'$.
Vertex $v_5$ is part of $C_1$ in $F$ and is bounded to $C$ in $F'$.
If $v$ is bounded  to a circuit $C_1$ or $C_2$ in $F$, then $v$ is bounded to $C$ in $F'$.
Therefore, if $F$ is BE factor, then so is $F'$ and the X-circuits of $F$ are identical 
to X-circuits of $F'$ except for the two X-circuits containing $C_1$ and $C_2$, which are merged into one.

On the other hand if $v_2v_3 \not\in F$, 
then $v_1v_2 \in F$, $v_3v_4 \in F$,
and $v_5$ is an isolated vertex of $F$.
Without loss of generality let $v_1v_2 \in C_1$ and $v_3v_4 \in C_2$.
We define $F'=(F-E(C_5)) \cup v_2v_3 \cup v_1v_5v_4$ (Figure~\ref{fig5swap}, right).
This operation merges $C_1$, $C_2$, and $v_5$ into one circuit $C$ 
and leaves the remaining circuits unaffected, thus $c(F)-c(F') = 3$. 
We call this operation a \emph{$5$-swap of type 2} on $C_5$ and we say that $C_1$ and $C_2$ \emph{participate in the swap}.
If $F$ is a BE factor, then the vertex $v_5$ is bounded to either $C_1$ or $C_2$, say it is $C_1$.
The vertex set  $V_F' = V_F - \{v_5\}$. 
If $v$ is bounded to a circuit $C_v\in {\cal C}_F - C_1 -C_2$ in $F$,
then $v$ is bounded to $C_v$ in $F'$.
If $v$, $v\neq v_5$, is bounded to a circuit $C_1$ or $C_2$ in $F$,
then $v$ is bounded to $C$ in $F'$.
The vertex $v_5$ is part of $C$ in $F'$.
Therefore, if $F$ is BE factor, then so is $F'$ and the X-circuits of $F$ are identical 
to X-circuits of $F'$ except for the two X-circuits containing $C_1$ and $C_2$, which are merged  into one.
A \emph{$5$-swap} is either a $5$-swap of type 1 or 2.
\begin{figure}[htp]
\center
\includegraphics[scale=1]{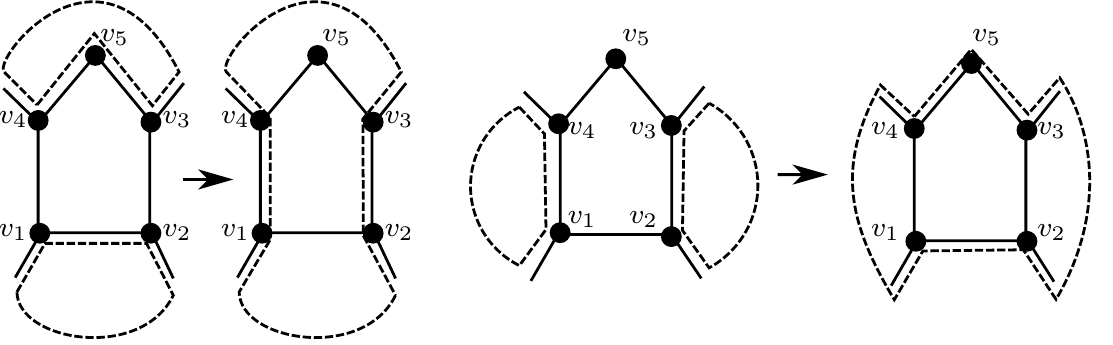}
\caption{A $5$-swap.}
\label{fig5swap}
\end{figure}

Let $C_6=v_1v_2v_3v_4v_5v_6$ be a $6$-circuit without a chord.
Assume that all edges from $\partial(C_6)$ belong to three distinct circuits of $F$, $C_1$, $C_2$, and $C_3$. 
Without loss of generality assume that $v_1v_2\in C_1$, $v_3v_4\in C_2$, 
and $v_5v_6\in C_3$.
We define $F'=(F-E(C_6)) \cup v_1v_6 \cup v_2v_3 \cup v_4v_5$ (Figure \ref{fig6swap}).
This operation merges the circuits $C_1$, $C_2$, and $C_3$ into one circuit $C$ in $F'$ 
and leaves the remaining circuits unaffected, thus $c(F)-c(F') = 4$. 
We call this operation a \emph{$6$-swap}  on $C_6$ and we say that $C_1$, $C_2$, and $C_3$ \emph{participate in the swap}.
The vertex set  $V_F' = V_F$. 
If $v$ is bounded to a circuit $C_v\in {\cal C}_F - C_1 -C_2 -C_3$ in $F$,
then $v$ is bounded to $C_v$ in $F'$.
If $v$ is bounded to a circuit $C_1$, $C_2$, or $C_3$ in $F$,
then $v$ is bounded to $C$ in $F'$.
Therefore, if $F$ is BE factor, then so is $F'$ and the X-circuits of $F$ are identical 
to X-circuits of $F'$ except for the three X-circuits containing $C_1$, $C_2$, and $C_3$ which are merged into one.
\begin{figure}[htp]
\center
\includegraphics[scale=1]{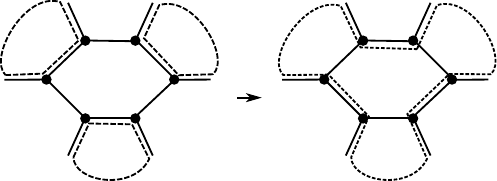}
\caption{A $6$-swap.}
\label{fig6swap}
\end{figure}

If an even factor $F$ is a BE factor, and circuits $C_1$ and $C_2$ (and $C_3$) participate in a swap,
then also the X-circuits $X_1$ containing $C_1$ and $X_2$ containing $C_2$ (and $X_3$ containing $C_3$)
\emph{participate in a swap}.
We say that an X-circuit $X_1$ of an even factor $F_1$ \emph{includes} an X-circuit $X_2$ of an even factor $F_2$ when $V(X_2) \subseteq V(X_1)$. If $X_1$ includes $X_2$
and $X_1$ participates in a swap we say that $X_2$ \emph{is part of the swap}.
We proved the following.
\begin{lemma} \label{obsswaps}
Let $F$ be a BE factor of a cubic graph $G$. Let $F'$ be an even factor obtained from $G$
by a $4$-swap, $5$-swap, or $6$-swap. Then
\begin{enumerate} 
\item $F'$ is a BE factor.
\item X-circuits of $F$ that do not participate in the swap are also X-circuits of $F'$,
X-circuits of $F$ participating in the swap are merged into one new X-circuit of $F'$ and the vertex set of the new X-circuit is the union of the vertex sets of the participating X-circuits.
\item $c(F)-c(F')\ge 1$; if the swap is not a $5$-swap of type 1, then $c(F)-c(F')\ge 2$;
and if the swap is a $6$-swap, then  then $c(F)-c(F')\ge 4$.
\end{enumerate} 
\end{lemma}

In order to perform the swaps we need the assumption that the boundary 
of the circuit contains edges from different circuits of an even factor.
The following lemma specifies when this assumption is true.
\begin{lemma}\label{swapposible}
Let $F$ be a BE factor of a cubic graph $G$. Let $X \in {\cal X}_F$ and $V_X=\emptyset$ (hence $X = C_X$).
Let $C$ be a circuit that intersects $X$ such that $V(C)-V(X) \neq \emptyset$. Then there is a circuit
$C' \in {\cal C}_F$ intersecting $C$ such that $C' \neq C_X$.
\end{lemma}
\begin{proof}
Assume that there is no circuit $C' \in {\cal C}_F$ intersecting $C$ such that $C' \neq C_X$.
As each vertex of $V(C)-V(X)$ has only one neighbour outside $V(C) \cup V(X)$
the vertices of $V(C)-V(X)$ cannot be bounded to a circuit other than $C_X$.
Since $F$ is a BE factor, the vertices of $V(C)-V(X)$ must be bounded to $C_X$, which contradicts $V_X=\emptyset$.
\end{proof}

\section{Finding a suitable $2$-factor}\label{sec4}
Consider a $4$-circuit of a $2$-factor that does not intersect any other circuits of length less than $7$.
Such $4$-circuits (and several other types of circuits) pose a problem as their cost is high 
(the circuit has cost $6$, thus the cost is $1.5$ per vertex) and it cannot be reduced by swapping operations.
We first give an intuition on how to deal with this problem. Let $G$ be a bridgeless cubic graph, let $F$ be a $2$-factor of $G$,
and let $C$ be the circuit mentioned above.
Each cut in $G$ contains even number of $2$-factor edges, 
which means that either 0,2, or 4 edges from $\partial(C)$ are in a $2$-factor.
"On average" a $2$-factor contains $2/3$ of the edges of $G$. Hence we could say that 
$F$ should contain "on average" $2.\bar{6}$ from $\partial(C)$.
This suggests that there should be a $2$-factor with all edges of $\partial(C)$ in it.
Assume that more circuits like $C$ are present in $G$.
We find a $2$-factor $F$ that contains as many edges from the boundaries of these circuits
as possible. We cannot guarantee that none of these circuits is in $F$,
but by the averaging argument each circuit with no boundary edges in $F$ should be compensated by 
at least two circuits with $4$ boundary edges in $F$. 
The $4$ boundary edges in $F$ can be either part of one circuit in $F$, which is relatively long (thus the cost per vertex is small),
or part of two circuits in $F$ with possible $4$-swap on the circuit $C$ (the swap decreases the cost).
This way we can compensate for the high cost of $4$-circuits in $F$.

Of course, $4$-circuits discussed above are not the only subgraphs 
that pose a problem. 
The aim of this section is to formalise the intuition given in the previous paragraph.
Although we use a different formalism, the arguments that follow can be reformulated probabilistically (thus the term "on average" used in the previous paragraph is appropriate), see section 3 of \cite{7-circuits} for details.

Let $G$ be a bridgeless cubic graph. 
A collection ${\cal H}$ of induced subgraphs is \emph{good collection} if each subgraph in ${\cal H}$
has the same number of vertices, denoted by $n_{\cal H}$, 
and the same number of boundary edges, denoted by $b_{\cal H}$,
$b_{\cal H} \in \{2, 4, 5, 6, 7, 9\}$ (although we do not use good collections with
$b_{\cal H} \in \{7, 9\}$ in our proof, we included them to present the technique in full strength). 
Let $F$ be a $2$-factor of $G$.
We denote ${\cal H}^k(F)$ the set of subgraphs from ${\cal H}$  that have $k$ boundary edges in $F$.
Note that $k$ must be even.
We define $a_{\cal H} =  2 \cdot  \lfloor b_{\cal H}/2\rfloor$.
Instead of ${\cal H}^{a_{\cal H}}(F)$ we will use ${\cal H}^*(F)$ and we omit $F$ if it is clear from the context. 
For a collection of subgraphs ${\cal H}$ let ${\cal \V H}$ be the union of the vertices in the subgraphs from ${\cal H}$.

\begin{lemma}
\label{circuit}
Let $G$ be a bridgeless cubic graph, 
let ${\cal I}$ be a set of good collections of induced subgraphs of $G$
and for each ${\cal H} \in {\cal I}$ let $A_{\cal H}$ be a real number.
There exists a $2$-factor $F$ that contains exactly two edges of each $3$-edge-cut and such that 
$$
\sum_{{\cal H} \in {\cal I}} A_{\cal H} \cdot
\left[ 
2 \cdot b_{\cal H} \cdot |{\cal H}^0| - 
(3\cdot a_{\cal H}-2 \cdot b_{\cal H}) \cdot |{\cal H}^*|
\right]
\le 0.
$$
\end{lemma}
\begin{proof}
Let $M$ be a $1$-factor of $G$ and let
$$
f(M)= \sum_{{\cal H} \in {\cal I}} A_{\cal H} \cdot \sum_{H\in {\cal H}} \sum_{e \in  \partial(H)} x_e,
$$
where $x_e=1$ if $e\in M$ and $x_e=0$ if $e\not\in M$. Let $F$ be the $2$-factor complementary to $M$. 
From the characterisation of the perfect matching polytope due to Edmonds \cite{edmonds}
we know that a fractional perfect matching, where each edge has value $1/3$, is a convex combination
of perfect matchings.  Among the perfect matchings in the convex combination we pick one with $f(M)$ minimal.
A perfect matching can contain only one or three edges of a $3$-edge-cut, therefore,
all perfect matchings in the convex combination contain exactly one edge from each $3$-edge-cut.
Thus $F$ contains exactly two edges from each $3$-edge-cut. Due to linearity of $f$
\begin{eqnarray}
f(M)\le \sum_{{\cal H} \in {\cal I}}
 A_{\cal H} \cdot \frac{b_{\cal H}}{3} \cdot |{\cal H}|
=
\sum_{{\cal H} \in {\cal I}} 
A_{\cal H} \cdot \sum_{j=0}^{a_{\cal H}/2} \frac{b_{\cal H}}{3}|{\cal H}^{2j}(F)|.
\label{eq1}
\end{eqnarray}
On the other hand, subgraphs from ${\cal H}^{k}$ have $b_{\cal H}-k$ boundary edges in $M$.
$$
f(M) = \sum_{{\cal H} \in {\cal I}}  A_{\cal H} \cdot \sum_{j=0}^{a_{\cal H}/2} (b_{\cal H}-2j) \cdot |{\cal H}^{2j}| 
$$
As $b_{\cal H}-2j \ge b_{\cal H}/3$, for integers $j$, $1\le j \le a_{\cal H}/2-1$ ($b_{\cal H}$ is
$2$, $4$, $5$, $6$, $7$, or $9$ and then $a_{\cal H}$ is $2$, $4$, $4$, $6$, $6$, or $8$, respectively) we get
\begin{eqnarray}
f(M)\ge
\sum_{{\cal H} \in {\cal I}} A_{\cal H} \cdot
\left( 
b_{\cal H} \cdot |{\cal H}^0| + 
(b_{\cal H}-a_{\cal H}) \cdot |{\cal H}^*| +
\sum_{j=1}^{a_{\cal H}/2-1}\frac{b_{\cal H}}{3} \cdot |{\cal H}^{2j}|
\right). \label{eq1b}
\end{eqnarray}
We combine (\ref{eq1}) and (\ref{eq1b}), multiply both sides of the inequality by $3$ and the lemma follows.
\end{proof}

\begin{lemma}\label{factorlemma}
Let $G$ be a bridgeless cubic graph, let $r$ be a real number, 
let ${\cal I}$ be a set of good collections of induced subgraphs of $G$.
For each ${\cal H}\in {\cal I}$, let $p_{\cal H}$, $s_{\cal H}$, and $t_{\cal H}$ be
three real numbers such that
\begin{eqnarray}
r \ge { 
\left(1.5 \cdot a_{\cal H}- b_{\cal H}\right)
\cdot s_{\cal H}  
+
\frac{p_{\cal H}}{n_{\cal H}}\cdot b_{\cal H}
\cdot t_{\cal H} 
\over
\left(1.5 \cdot a_{\cal H}- b_{\cal H}\right)
+
\frac{p_{\cal H}}{n_{\cal H}}\cdot b_{\cal H}
}. \label{eqinlemma}
\end{eqnarray}
Then there exists a $2$-factor $F$ that contains an edge of each $3$-edge-cut and such that for each function $c: V(G) \to \mathbb{R}$ satisfying
\begin{itemize}
\item for each ${\cal H} \in {\cal I}$ and for each $v \in {\cal \V H}^0$, $c(v)\le s_{\cal H}$;
for all other vertices $v$, $c(v)\le r$ and
\item for each ${\cal H} \in {\cal I}$ there exist a set ${\cal \SV H}$ 
of size at least $p_{\cal H} \cdot |{\cal H}^*|$ such that
for each $v\in {\cal \SV H}$, $c(v) \le t_{\cal H}$, and these sets are pairwise disjoint,
\end{itemize}
$$\sum_{v\in V(G)} c(v) \le r \cdot |V(G)|.$$
\end{lemma}
\begin{proof}
As all vertices outside ${\cal \V H}^0$, for all ${\cal H} \in {\cal I}$, have cost at most $r$,
to prove the lemma it suffices to prove
$$
\sum_{{\cal H} \in {\cal I}} \sum_{v \in {\cal \V H}^0} (c(v) - r) \le
\sum_{{\cal H} \in {\cal I}} \sum_{v \in {\cal \SV H}} (r-c(v)).
$$
Due to the bounds on the costs of the vertices and the bound on the size of ${\cal \SV H}$ this is implied by
\begin{eqnarray}
\sum_{{\cal H} \in {\cal I}}  (s_{\cal H}-r) \cdot |{\cal \V H}^0|- 
                              (r-t_{\cal H}) \cdot |{\cal \SV H}| &\le& 0 \nonumber \\
\sum_{{\cal H} \in {\cal I}} n_{\cal H} \cdot (s_{\cal H}-r) \cdot |{\cal H}^0| - 
                             p_{\cal H} \cdot (r-t_{\cal H}) \cdot |{\cal H}^*| &\le& 0. \label{eqres2}
\end{eqnarray}
Equation (\ref{eqinlemma}) can be rewritten as
\begin{eqnarray}
r \ge { 
\frac{n_{\cal H}}{2 \cdot b_{\cal H}}
\cdot s_{\cal H}  
+
\frac{p_{\cal H}}{\left(3\cdot a_{\cal H}- 2 \cdot b_{\cal H} \right)}
\cdot t_{\cal H} 
\over
\frac{n_{\cal H}}{2 \cdot b_{\cal H}}
+
\frac{p_{\cal H}}{\left(3\cdot a_{\cal H}- 2 \cdot b_{\cal H} \right)}
}. \nonumber
\end{eqnarray}
and
\begin{eqnarray}
\frac{p_{\cal H}}{\left(3 \cdot a_{\cal H}-2 \cdot b_{\cal H}\right)} \cdot
(r-t_{\cal H} )
\ge  
\frac{n_{\cal H}}{2 \cdot b_{\cal H}}
\cdot (s_{\cal H} - r). \nonumber
\end{eqnarray}
Thus by manipulating (\ref{eqres2}) it suffices to prove
$$
\sum_{{\cal H} \in {\cal I}} p_{\cal H} \cdot \frac{2 \cdot b_{\cal H}}{\left( 3\cdot a_{\cal H}- 2 \cdot b_{\cal H}\right)} 
\cdot (r-t_{\cal H}) \cdot |{\cal H}^0| - 
                             p_{\cal H} \cdot (r-t_{\cal H}) \cdot |{\cal H}^*| \le 0.
$$
This equation is implied by Lemma~\ref{circuit} when we set
$$
A_{\cal H}=\frac{p_{\cal H}}{\left(a_{\cal H}-\frac{2b_{\cal H}}{3}\right)} \cdot (r-t_{\cal H}).
$$
\end{proof}
In our proofs ${\cal \SV H}\subseteq {\cal \V H}^*$ and the ratio
$p_{\cal H} / n_{\cal H}$ expresses how many vertices
of ${\cal H}^*$ per subgraph are in ${\cal \SV H}$.
This ratio will be used when taking into account that the subgraphs in ${\cal H}^*$
may not be disjoint or low value of $c$ may not be guaranteed on all vertices from ${\cal \V H}^*$.
The expression (\ref{eqinlemma}) is a weighted average of $s_{\cal H}$ and $t_{\cal H}$, where many terms 
depend only on $b_{\cal H}$. 
To make the formula more readable we introduce the terms $u_{\cal H}$ and $v_{\cal H}$ and rewrite the formula as
$$
r \ge \frac{u_{\cal H} \cdot s_{\cal H} + v_{\cal H} \cdot t_{\cal H}}{u_{\cal H} + v_{\cal H}},
$$
where $u_{\cal H}$ and $v_{\cal H}$ are defined depending on $b_{\cal H}$ as follows.
\begin{center}
\begin{tabular}{ccc}
$b_{\cal H}$ & $u_{\cal H}$ & $v_{\cal H}$ \\
$2$ & $1$ & $p_{\cal H} / n_{\cal H} \cdot 2$\\
$4$ & $1$ & $p_{\cal H} / n_{\cal H} \cdot 2$\\
$5$ & $1$ & $p_{\cal H} / n_{\cal H} \cdot 5$\\
$6$ & $1$ & $p_{\cal H} / n_{\cal H} \cdot 2$\\
$7$ & $2$ & $p_{\cal H} / n_{\cal H} \cdot 7$\\
$9$ & $1$ & $p_{\cal H} / n_{\cal H} \cdot 3$
\end{tabular}
\end{center}

\section{Performing swaps}\label{sec5}

In this section we define which swapping operations are used,
in which order are they used, and how the cost of the vertices (used as the function $c$ in Lemma~\ref{factorlemma}) is calculated.
Let $G$ be a bridgeless irreducible cubic graph.
 We define ${\cal C}^*$ to be a set that contains all induced circuits of length $4$ and $5$,
that is a set of circuits of $G$ where $4$- and $5$-swaps could be made.
We say that a circuit $C$ \emph{touches} an induced  subgraph $H$ of $G$ if $V(C) \cap V(H) \neq \emptyset$  and $V(C) - V(H) \neq \emptyset$. Note that by Lemma~\ref{swapposible}, if a circuit $C$ from $C^*$ touches 
a circuit of a $2$-factor $F$ of $G$, then $4$- or $5$-swap on $C$ can be made to merge two circuits
of $F$ touching $C$.

Let $F$ be a triangle-free $2$-factor of $G$. 
Swapping operations on $F$ are carried out in two phases.

\medskip
\noindent
\emph{Phase 1:} We do swapping operation 
in the following order of preference until no longer possible:
\begin{enumerate} 
\item $6$-swaps on circuits that do not touch a circuit from ${\cal C}^* \cap F$,
\item $4$-swaps on circuits that do not touch a circuit from ${\cal C}^* \cap F$, 
\item $5$-swaps on circuits that do not touch a circuit from ${\cal C}^* \cap F$.
\end{enumerate}
\medskip

\begin{observation}
\label{obsswap}
In phase 1 only circuits of length at least $6$ can participate in swapping operations.
\end{observation}

As we start with a BE factor ($2$-factor is by definition a BE factor),
by Lemma~\ref{obsswaps} part 1, we get a new BE factor $F_1$ at the end of the phase 1. 
We calculate the cost of each vertex, denoted by $c_1$, as follows.
If $v$ is in an X-circuit $X$ of $F_1$ that contains no vertices from a 4-diamond, then $c_1(v)=c(X)/|V(X)|$ (recall that $c(X) = |V(X)| + |V_X| + 2$).
Let $X$ be an X-circuit that contains $k$ vertices in 4-diamonds. 
We first set $c_1(v)=c(X)/|V(X)|$ for each vertex $v\in V(C)$. 
If $c_1(v)\le 1.2$ for each $v\in V(C)$, then we are done.
If $c_1(v)> 1.2$, then vertices in 4-diamonds get cost $1.2$ and
vertices outside 4-diamonds get cost $(c_1(X)-1.2\cdot k)/(|V(X)|-k)$.
Note that $\sum_{v\in V(G)} c_1(v) = c(F_1)$.

We start by examining the cost of vertices after the first phase. 
Each X-circuit $X\in F_1$ is either a circuit of $F$ or is 
created from circuits of $F$ using $6$-swaps, $4$-swaps, or $5$-swaps.
\begin{lemma}\label{hlema}
Let $G$ be an irreducible bridgeless cubic graph and let $F$ be a  $2$-factor.
Let $F_1$ be the BE factor obtained after performing all swaps in phase 1.
Let $X$  be an X-circuit of $F_1$. If $X$ is created from circuits of $F$
by $j_k$ $k$-swaps, for $k\in \{4, 5, 6\}$, $j_4+j_5+j_6 \neq 0$,  
$X$ contains vertices of $j_d$ diamonds and $j_i$ isolated vertices, then 
\begin{enumerate}
\item $X$ contains either none or all vertices of each $4$-diamond
\item $|V(X)| \ge 6\cdot(1+j_4+j_5+2\cdot j_6)+2\cdot j_d$
\item $j_i\le j_5$ 
\item $|V(X)|-4 \cdot j_d > 0$
\item For each $v\in V(X)$
$$
c(v) \le \max \left\{ 1.2, 1+ \frac{2+j_5-0.8 \cdot j_d}{|V(X)|-4 \cdot j_d} \right\} 
$$
\item If $0 \le -0.8 + 0.2\cdot j_5 + 0.4 \cdot j_d$, then for each 
$v\in V(X)$ $c(v) \le 1.2$.
\end{enumerate}
\end{lemma}
\begin{proof}
As all circuits of $F$ contain either zero or all vertices of each $4$-diamond,
by Lemma~\ref{obsswaps} part 2, so do X-circuits of $F_1$. Thus part 1 of the lemma is true.

Consider a circuit $C$ of $F$ contained in an X-circuit of $F_1$.
By Observation~\ref{obsswap}, $C$ has length at least $6$.
If $C$ contains vertices of exactly one 4-diamond, then
the length of $C$ is at least $8$: 
if it was $6$, then $C$ would be in a 6-diamond, not a 4-diamond (see the definition of 4-diamond) and
if it was $7$, then $G$ would contain a type 3 reducible
configuration; neither of these two situations can happen.
If $C$ contains vertices of two 4-diamonds, then if it has length $8$, then the circuit is Hamiltonian and no swaps were made (contradicting  $j_4+j_5+j_6 \neq 0$). The circuit $C$ cannot have length $9$ because $G$ is bridgeless, so $|V(C)| \ge 10$.
If $C$ contains vertices of $k$ diamonds where $k>2$, then $|V(C)| \ge 4\cdot k$.
For any $k\geq 0$ if $C$ contains vertices of $k$ 4-diamonds, then $|V(C)| \ge 6+2\cdot k$.
The $4$- and $5$-swap merge two X-circuits and the $6$-swap merges three X-circuits.
By Lemma~\ref{obsswaps} part 2, part 2 of this lemma follows. 

Since only 5-swaps can create isolated vertices and each 5-swap creates 
only one of them, the part 3 of this lemma follows.

Part 4 of the lemma follows from $j_4+j_5+j_6 \neq 0$ and from part 1 of this lemma. Hence the formula in part 5 is defined and follows from the definition of $c_1$ and part 3 of this lemma.

Finaly, we prove part 6 of the lemma.
By part 2 of Lemma~\ref{hlema}, $|V(X)| \ge 6 \cdot j_5 + 6 + 2 \cdot j_d$,
and by part 4 of Lemma~\ref{hlema}, $|V(X)|-4\cdot j_d > 0$. But then
\begin{eqnarray*}
1+ \frac{2+j_5-0.8 \cdot j_d}{|V(X)|-4 \cdot j_d} &\le& 1.2 \\
2+j_5-0.8 \cdot j_d &\le& 0.2 \cdot \left( |V(X)|-4 \cdot j_d \right) \\
2+j_5-0.8 \cdot j_d &\le& 0.2 \cdot \left( 6 + 6\cdot j_5 -2\cdot j_d \right) \\
0 &\le& -0.8 + 0.2\cdot j_5 + 0.4 \cdot j_d \\
\end{eqnarray*}
and by part 5 of Lemma~\ref{hlema}, $c(v) \le 1.2$.
\end{proof}

Using this lemma we can bound the costs of the vertices in X-circuits of $F_1$ 
obtained by swapping operations.

\begin{corollary}
\label{swapcost}
Let $X$ be an X-circuit of $F_1$ and let $v\in V(X)$. 
If at least one $6$-swap or $4$-swap was used to create $X$, then $c_1(v) \le 1.2$.
\end{corollary}
\begin{proof}
By part 2 of Lemma~\ref{hlema}, $|V(X)| \ge 6 \cdot j_5 + 12 + 2 \cdot j_d$, 
and by part 4 of Lemma~\ref{hlema}, $|V(X)|-4\cdot j_d > 0$. But then
\begin{eqnarray*}
1+ \frac{2+j_5-0.8 \cdot j_d}{|V(X)|-4\cdot j_d} &<& 1.2 \\
2+j_5-0.8 \cdot j_d &<& 0.2 \cdot \left(|V(X)|-4\cdot j_d \right) \\
2+j_5-0.8 \cdot j_d &<& 0.2 \cdot \left(12 + 6\cdot j_5 - 2\cdot j_d \right) \\
0 &<& 0.4 + 0.2\cdot j_5 + 0.4\cdot j_d, 
\end{eqnarray*}
and by part 5 of Lemma~\ref{hlema}, $c(v) \le 1.2$.
\end{proof}

\begin{corollary}
\label{swapcost4}
Let $X$ be an X-circuit of $F_1$ and let $v\in V(X)$. 
If at least four $5$-swaps were used to create $X$, then $c_1(v) \le 1.2$.
\end{corollary}
\begin{proof}
The corollary follows by part 6 of Lemma~\ref{hlema}.
\end{proof}

\begin{corollary}
\label{swapcost3}
Let $X$ be an X-circuit of $F_1$ such that exactly three $5$-swaps were used to create $X$.
Then for each $v \in X$  $c_1(v) \le 1.2$ except for the case when $|V(X)| = 24$ and exactly three $5$-swaps were used to create $X$. In the latter case $c_1(v) \le 29/24 < 1.21$.
\end{corollary}
\begin{proof}
If $j_d=0$, then by part 5 of Lemma~\ref{hlema}, the corollary follows. 
If $j_d>0$, then by part 6 of Lemma~\ref{hlema}, the corollary follows. 
\end{proof}

\begin{corollary}
\label{swapcost2}
Let $X$ be an X-circuit of $F_1$ such that exactly two $5$-swaps were used to create $X$.
Then for each $v \in X$: 
\begin{itemize}
\item if $|V(X)|\ge 20$, then $c_1(v) \le 1.2$,
\item otherwise $c_1(v) \le 11/9$.
\end{itemize}
\end{corollary}
\begin{proof}
If $j_d=0$, then by part 5 of Lemma~\ref{hlema}, the corollary follows. 
If $j_d>0$, then by part 6 of Lemma~\ref{hlema}, the corollary follows. 
\end{proof}

\begin{corollary}
\label{swapcost1}
Let $X$ be an X-circuit of $F_1$ such that exactly one $5$-swap was used to create $X$.
Then for each $v \in X$: 
\begin{itemize}
\item if $|V(X)|\ge 15$, then $c_1(v) \le 1.2$;
\item otherwise, $c_1(v) \le 1.25$.
\end{itemize}
\end{corollary}
\begin{proof}
If $j_d=0$ or $j_d=1$, then by part 5 of Lemma~\ref{hlema}, 
in each case we can calculate the costs of the vertices and the corollary follows. 
If $j_d>1$, then by part 6 of Lemma~\ref{hlema}, the corollary follows. 
\end{proof}

\begin{corollary}
\label{swapcost0}
Let $X$ be an X-circuit of $F_1$ such that no swaps were used to create $X$.
Then for each $v \in X$: 
\begin{itemize}
\item if $|V(X)|\ge 10$, then $c_1(v) \le 1.2$,
\item if $|V(X)|\ge 9$, then $c_1(v) \le 1.24$,
\item if $|V(X)|\ge 8$, then 
$c_1(v) \le 1.25$ when $X$ does not intersect a 4-diamond and
$c_1(v) \le 1.3$ when $X$ does intersect a 4-diamond,
\item if $|V(X)|= 7$, then $c_1(v) \le 9/7$,
\item if $|V(X)|\ge 6$, then $c_1(v) \le 4/3$,
\item if $|V(X)|\ge 5$, then $c_1(v) \le 1.4$;
\item otherwise, $c_1(v) \le 1.5$.
\end{itemize}
\end{corollary}
\begin{proof}
The statement implies from the definition of $c_1$.
\end{proof}

Phase 2 does not take into account whether a swap is made on a circuit that touches a circuit from  ${\cal C}^* \cap F$. Moreover, the cost will be calculated differently.
\medskip
\noindent
\emph{Phase 2:}  
We do swapping operations until no longer possible
in the following order of preference:
\begin{enumerate} 
\item $6$-swaps,
\item $4$-swaps, 
\item $5$-swaps.
\end{enumerate}
Lemma~\ref{obsswaps} part 3 shows that the cost of the even factor will decrease by some amount $s$ after swap.
Assume that $m$ X-circuits $X_i$, \ldots,$X_m$ participate in a swap ($m\in \{2,3\}$). We decrease the cost
of vertices of $X_i$ which are not in a 4-diamond by $s/(m \cdot t)$, where $t$ is the number of 
vertices of $X_i$ outside 4-diamonds.
Let $F_2$ be the resulting even factor after phase 2. By Lemma~\ref{obsswaps} part 1, $F_2$ is a BE factor and we denote the costs of the vertices after phase 2 by $c_2$.
Note that $\sum_{v\in V(G)} c_2(v) = c(F_2)$.

\medskip

If swaps with short circuits
were performed in phase 1, then the costs of the vertices could increase beyond $1.2$ uncontrollably.
This is the reason, why swaps are separated into two phases. 
No such thing can happen for swaps performed in phase 2.
\begin{observation}\label{phase2}
The costs of the vertices can only decrease in phase 2.
\end{observation}
We bound the costs of the vertices in circuits of $F$ after phase 2 using the costs after
phase 1, which are bounded by Observations \ref{swapcost}--\ref{swapcost0}.

\medskip

Let ${\cal C}_4$, ${\cal C}_5$, and  ${\cal C}_6$ be sets of all subgraphs
induced by the vertex-set of a chordless $4$-, $5$-, and $6$-circuits of $G$ with
an independent boundary.
We define the following sets of induced subgraphs of $G$
that are vertex-sets of circuits of length at most $6$.
Note that each set forms a good collection.
\begin{itemize}
\item ${\cal D}_4$ : 4-diamonds

\item ${\cal D}_6$ : 6-diamonds

\item ${\cal C}_{\text{4-noint}}$ :  circuits from
${\cal C}_4$ that do not touch any circuit from ${\cal C}^*$

\item ${\cal C}_{\text{5-noint}}$ : circuits from
${\cal C}_5$ that do not touch any circuit from ${\cal C}^*$

\item ${\cal C}_{\text{6-noint}}$ : circuits from  ${\cal C}_6$ that do not touch any circuit from ${\cal C}^*$

\item ${\cal C}_{\text{4-4-noint}}$ : 
induced subgraphs on $6$ vertices and
$7$ edges that contain a $6$-circuit and two $4$-circuits intersecting in exactly one edge and 
such that they do not touch any circuit from ${\cal C}^*$

\item ${\cal C}_{\text{4-int-5}}$ : circuits from ${\cal C}_4$ that touch a $5$-circuit from ${\cal C}^*$ but do not touch a circuit of length $4$ or $6$ from ${\cal C}^*$.
\end{itemize}
Moreover, let
$$
{\cal I}=\left\{ {\cal D}_4, {\cal D}_6, {\cal C}_{\text{4-noint}},
{\cal C}_{\text{5-noint}}, {\cal C}_{\text{4-4-noint}}, 
{\cal C}_{\text{6-noint}}, {\cal C}_{\text{4-int-5}} \right\}.
$$
Although, we work with factors $F$, $F_1$, and $F_2$ in this chapter,
only one of these factors is guaranteed to be a $2$-factor: $F$.
Thus according to the agreement from the definition we will drop "$(F)$" 
in the notations of
${\cal H}^k(F)$, ${\cal \V H}^0(F)$, and in several other similar notations.

\begin{lemma}\label{lemmaall}
Let $G$ be a cubic bridgeless irreducible graph and let $F$ be a $2$-factor of $G$. If we do swapping operations on $F$ in the specified order, then after phase 2 for  each $v\in V(G)$ such that $v\not \in {\cal \V H}^0$ for each ${\cal H} \in {\cal I}$, $c_2(v) \le 1.3$.
\end{lemma}
\begin{proof}
Let $C$ be a circuit of $F$ and let $v\in V(C)$.
According to Corollaries~\ref{swapcost}--\ref{swapcost0} if 
$C$ participates in a swap in phase 1 or $|V(C)|\ge 7$, then
$c_1(v)\le 1.3$ and by Observation~\ref{phase2} also $c_2(v) \le 1.3$.
Thus we can suppose that $|V(C)| \in \{4, 5, 6\}$ and $C$ is an X-circuit of the BE factor
$F_1$ obtained after phase 1. 

Assume first that $C$ is a circuit of length $4$ with a chord.
As $G$ is irreducible, $C$ is either in a 4-diamond or in a $6$-diamond.
If $C$ is in a 4-diamond then $v \in {\cal \V \cal D}_4^0$, a contradiction.
If $C$ is in a 6-diamond then it touches a $4$-circuit that is inside this 6-diamond.
By Lemma~\ref{swapposible},  
a $4$-swap that takes $C$ as one participating circuit can be performed.
The $4$-swap will be indeed performed in phase 2 as there is no other way to merge $C$ with other X-circuits by a different swap. 
Thus the cost of $v$ decreases from
$1.5$ to $1.25$ by this swap (by Lemma~\ref{obsswaps} part 3, the cost is reduced by at least 2, and half of the reduction is divided equally among vertices of $C$)
and later the cost can only decrease further.
Thus the cost of $v \in V(C)$ is at most $1.3$.

If $C$ is a $4$-circuit without a chord, then it is touched by a $4$-circuit
(otherwise either $v\in {\cal \V C}_{\text{4-noint}}^0$ or 
$v \in {\cal \V C}_{\text{4-int-5}}^0$).
By Lemma~\ref{swapposible},  a $4$-swap can be performed in
 phase 2 such that $C$ participates in it. 
According to the preference of swaps defined for phase 2 either a $6$-swap,
or a $4$-swap will be performed in phase 2. By Lemma~\ref{obsswaps} part 3,
the spared cost is at least $4$ or $2$, 
and this spared cost will be equally distributed among $3$ or $2$ X-circuits,
respectively.
In each case the total cost of the vertices of $C$ is reduced by at least $1$
thus the cost of each vertex is reduced 
by at least $0.25$. As $c_1(v) =1.5$, $c_2(v) \le 1.25 < 1.3$.

If $C$ is a $5$-circuit, then it is chordless because $G$ is irreducible and thus does not 
contain a type 1 reducible subgraph.
It also touches a circuit from ${\cal C}^*$, otherwise
$v\in {\cal \V C}_{\text{5-noint}}^0$.
By Lemma~\ref{swapposible}, a swap  
can be performed in phase 2 such that $C$ participates in it.
Thus some swap such that $C$ participates in it will be performed in phase 2.
The total cost of the vertices of $C$ will be decreased by at least $0.5$
and hence the cost of $v \in V(C)$ is decreased from $1.4$ after phase $1$ to
at least $1.3$ after phase 2.

If $C$ is a $6$-circuit with two chords, then $v\in {\cal \V D}_6^0$. 
If $C$ is a $6$-circuit with one chord, then the chord 
connects the vertices of $C$ in distance $3$ because $G$ is irreducible.
Then $C$ has to touch a circuit from ${\cal C}^*$,
otherwise $v\in {\cal \V C}_{\text{4-4-noint}}^0$. 
By Lemma~\ref{swapposible}, a swap can be performed 
in phase 2 such that $C$ participates in it.
And thus some swap will be performed in phase 2.
This swap decreases the cost of $v$ from $4/3$ to at most
$1.25$, which is less than $1.3$

Finally, if $C$ is a $6$-circuit without a chord, then $C$ touches a circuit from ${\cal C}^*$,
otherwise $v\in {\cal \V C}_{\text{6-noint}}^0$. 
By Lemma~\ref{swapposible}, a swap can be performed 
in phase 1 or 2 such that $C$ participates in it.
If such a swap was performed in phase 1, then $c_1(v) \le 1.25$
by Observations~\ref{swapcost}--\ref{swapcost1} and thus $c_2(v)\le 1.25$.
If no swaps were performed in phase 1, then $C$ must participate in a swap in phase 2. This swap decreases the cost of $v$ from $4/3$ to at most
$1.25$, which is less than $1.3$.
\end{proof}

Now we prove statements on the costs of vertices from ${\cal \V H}^0$
and ${\cal \V H}^*$, for ${\cal H} \in {\cal  I}$.
\begin{lemma}\label{lemma0}
Let $G$ be a cubic bridgeless irreducible graph, let $F$ be a $2$-factor of
$G$. If we do swapping operations on $F$ in the specified order, then after phase 2 for each
$v \in {\cal \V D}_4^0$, $v\in {\cal \V D}_6^0$, $v\in {\cal\V  C}_{\text{4-noint}}^0$,
$v\in {\cal \V C}_{\text{5-noint}}^0$, $v \in{\cal \V C}_{\text{4-4-noint}}^0$, 
$v \in {\cal \V C}_{\text{6-noint}}^0$, and $v\in {\cal \V C}_{\text{4-int-5}}^0$
we have 
$c_2(v) \le$ $1.5$, $4/3$, $1.5$, $1.4$, $4/3$, $4/3$, and $1.375$, respectively.
\end{lemma}
\begin{proof}
For $v \in {\cal \V D}_4^0$, $v\in {\cal \V D}_6^0$, $v\in {\cal\V  C}_{\text{4-noint}}^0$,
$v\in {\cal \V C}_{\text{5-noint}}^0$, $v \in{\cal \V C}_{\text{4-4-noint}}^0$, and
$v \in {\cal \V C}_{\text{6-noint}}^0$
the statement is implied by Observation~\ref{obsswap} and Corollary~\ref{swapcost0}.

If $v\in {\cal \V C}_{\text{4-int-5}}^0$, then 
by Observation~\ref{obsswap} and Lemma~\ref{swapposible}, 
a $5$-swap is available on the $5$-circuit that touches $C$. Any swap 
that $C$ participates in can be performed only in phase 2.
By Lemma~\ref{obsswaps}, the cost of $v$ is decreased from $1.5$ by at least $1/8$,
that is to at most $1.375$.
\end{proof}

\begin{lemma}\label{lemmastar}
Let $G$ be a cubic bridgeless irreducible graph, let $F$ be a $2$-factor of
$G$. If we do swapping operations on $F$ in the specified order, then after phase 2 for 
$v \in {\cal \V D}_4^*$, $v\in {\cal \V D}_6^*$, $v\in {\cal\V  C}_{\text{4-noint}}^*$,
$v\in {\cal \V C}_{\text{5-noint}}^*$, $v \in{\cal \V C}_{\text{4-4-noint}}^*$, 
and $v\in {\cal \V C}_{\text{4-int-5}}^*$
we have 
$c_2(v) \le$ $1.2, 1.25, 1.2, 1.25, 1.2$, and $1.25$, respectively
\end{lemma}
\begin{proof}
Let $F_1$ be the BE factor obtained after phase 1.
The vertices of 4-diamonds get cost at most $1.2$ after phase 1 by definition.
Thus if $v \in {\cal \V D}_4^*$, then $c_2(v) \le 1.2$.
Consider $v \in {\cal \V D}_6^*$ and let
$H\in {\cal D}_6^*$ such that $v \in v(H)$.
Either exactly one or exactly two circuits of $F$ intersect $H$.

Assume that $F$ contains only one circuit $C$ that intersects $H$.
The length of $C$ is at least $9$ (otherwise $G$ would have a type 2 reducible configuration or a bridge) and
if $C$ contains vertices of a diamond, then $|V(C)| \ge 10$. By Observation~\ref{obsswap},
if $C$ is part of $1$, $2$, or $3$ swaps  preformed in phase 1,
then the number of vertices of the X-circuit of $F_1$ that includes $C$ is at least $15$, $21$, and $27$.  By Corrolaries~\ref{swapcost}--\ref{swapcost0}, $c_1(v)\le 1.25$,
and by Observation~\ref{phase2}, $c_2(v)\le 1.25$.

On the other hand, if two circuits $C_1$ and $C_2$ of $F$ intersect $H$, one of them, say $C_2$ is of length $4$.
By Observation~\ref{obsswap} and Lemma~\ref{obsswaps} part 3, 
$C_2$ does not change in phase 1
and in phase 2 the cost of its vertices is decreased by at least $0.25$, thus for $v \in V(C_2)$, $c_2(v) \le 1.25$.
After phase 1 by Corollaries~\ref{swapcost}--\ref{swapcost0},
the X-circuit $X$ that includes $C_1$ either has $c_1(v)\le 1.25$ for each $v\in V(X)$,
or $C_1=X$ and $|V(X)| \le 8$ (Corollary~\ref{swapcost1}). 
In the first case the cost is at most $1.25$ also after phase 2. 
In the second case $X$ will participate in a $4$-swap or a $6$-swap in phase 2 (a $4$-swap is possible to merge $X$ with $C_2$).
If $X$ intersects no $4$-diamond, then the swap decreases the cost of vertices, $(|V(X)|+2)/|V(X)|$, by at least $1/|V(X)|$, thus the resulting cost is at most $1.25$.
If $X$ intersects a diamond (thus $|X|=8$), then the cost of vertices
outside the 4-diamond drops from $1.3$ (Corollary~\ref{swapcost0}) 
after phase 1 to at most $1.05$ after phase 2.

\begin{figure}[htp]
\center
\includegraphics[scale=1.5]{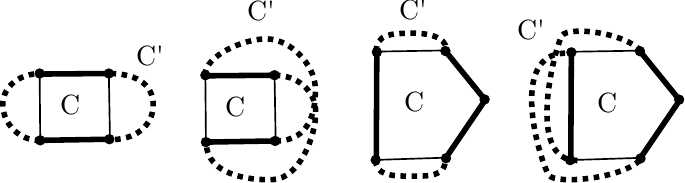}
\caption{Possible situations when all boundary edges of 
$C \in {\cal C}^*_{\text{4-noint}} \cup {\cal C}^*_{\text{5-noint}} \cup {\cal \V C}_{\text{4-int-5}}^*$ belong to one circuit of~$F$.
All dotted paths contain at least three inner vertices if 
$C \in {\cal C}^*_{\text{4-noint}} \cup {\cal C}^*_{\text{5-noint}}$ 
and at least two inner vertices if $C\in {\cal \V C}_{\text{4-int-5}}^*$.}
\label{figC45}
\end{figure} 
Assume now that $v\in {\cal \V C}_{\text{4-noint}}^*$ ($v\in {\cal \V C}_{\text{5-noint}}^*$) and let 
$C \in {\cal C}_{\text{4-noint}}^*$ ($C \in {\cal C}_{\text{5-noint}}^*$) such that  $v \in V(C)$.  
If only one circuit $C'$ of $F$ intersects $C$, then $C$ is in one of two positions from 
Figure~\ref{figC45}. Since $C$ has  an independent boundary and 
does not touch circuits of length $4$ and $5$, $|V(C')|\ge 10$ (|$V(C')|\ge 11$), see Figure~\ref{figC45}.
By Observation~\ref{obsswap}, if $C'$ is part of $0$, $1$, $2$, and $3$ $5$-swaps
then the number of vertices of the X-circuit $X$ of $F_1$ that includes $C'$ is at least $10, 16, 22, 28$ ($11,17,13,28$) edges, respectively. 
By Observations~\ref{swapcost}--\ref{swapcost0}, $c_1(v) \le 1.2$ in each case and thus $c_2(v) \le 1.2$. 
If two circuits $C_1$ and $C_2$ of $F$ intersect $C$, then it is possible to make a 
$4$-swap ($5$-swap) on $C$ and because $C$ does not touch circuits from 
${\cal C}^* $ the swap can be performed in phase 1.
Due to priorities  of swaps in phase 1 a $4$- or a $6$-swap (or a $5$-swap) will be performed on both $C_1$ and $C_2$. By Corollary~\ref{swapcost} (Corollaries~\ref{swapcost}--\ref{swapcost1}), $c_1(v)\le 1.2$ ($1.25$) and thus $c_2(v)\le 1.2$ ($1.25$).

Assume now that $v \in{\cal \V C}_{\text{4-4-noint}}^*$, 
and let $H \in {\cal C}_{\text{4-4-noint}}^*$ such that  $v \in V(H)$.
Suppose that only one circuit $C'$ of $F$ intersects $H$. As the subgraph $H$ has an independent boundary and does not touch circuits of length $4$ and $5$, 
$|V(C')|\ge 10$ (six vertices within $H$ and two vertices of both circuit segments outside $H$).
By Observation~\ref{obsswap}, the number of vertices of an X-circuit $X$ of $F_1$ that includes $C'$
is at least $10, 16, 22, 28$ when $C$ is part of at least $0$, $1$, $2$, and $3$ $5$-swaps in phase 1,
respectively. By Observations~\ref{swapcost}--\ref{swapcost0}, $c_1(v) \le 1.2$ in each case and thus $c_2(v) \le 1.2$. 
If two circuits $C_1$ and $C_2$ of $F$ intersect $H$, then it is possible to make a 
$4$-swap  on one of the $4$-circuits of $H$ and as $H$ does not touch circuits from 
${\cal C}^* $ the swap can be performed in phase 1.
Due to priorities  of swaps in phase 1 a $4$- or a $6$-swap will be performed on both $C_1$ and $C_2$. By Corollary~\ref{swapcost} , $c_1(v)\le 1.2$  and thus $c_2(v)\le 1.2$.

Finally assume that $v\in {\cal \V C}_{\text{4-int-5}}^*$,
and let $C \in {\cal C}_{\text{4-int-5}}^*$ such that  $v \in V(C)$.
Suppose that only one circuit $C'$ of $F$ intersects $C$. The circuit  $C$ is chordless, 
with independent boundary, and does not touch circuits of length $4$, 
thus $|V(C')|\ge 8$ if $C$ contains no vertices of a $4$-diamond and $|V(C')|\ge 10$ otherwise
(Figure~\ref{figC45}, left). By Observation~\ref{obsswap}, the X-circuit $X$ of $F_1$ that includes $C'$
has at least $8, 14, 20, 26$ vertices when $C'$ is part of  least $0$, $1$, $2$, and $3$ $5$-swaps in phase 1,
respectively. By Observations~\ref{swapcost}--\ref{swapcost0}, $c_1(v) \le 1.25$ in each case and thus $c_2(v) \le 1.25$. 
Suppose now that  two circuits $C_1$ and $C_2$ of $F$ intersect $H$.
If neither $C_1$ nor $C_2$ is a $5$-circuit then the argumentation from the case
$v\in {\cal \V C}_{\text{4-noint}}^*$ applies.
Otherwise, let $C_2$ be a $5$-circuit. 
By Observation~\ref{obsswap}, $C_2$ is an X-circuit of $F_1$ and the cost of vertices in
$C_1$ is at most $1.25$ unless $C_1$ is a circuit of length at most $7$
that does not intersect a 4-diamond, or a circuit of length $8$ that intersects a diamond (Observations~\ref{swapcost}--\ref{swapcost0}). In such cases $C_1$
is an X-circuit of $F_1$ too. In phase 2 due to the order in which swaps are
performed a $4$- or a $6$-swap was performed on both $C_1$ and $C_2$.
After phase 1, the vertices of the circuit $C_i$, $i\in \{1, 2\}$ and $|V(C_i)|<8$, have cost
$1+2 / |V(C_i)|$, which decreases to  $1+1 / |V(C_i)|$ with the first swap in phase 2, thus the vertices of $C_i$ have cost at most $1.2$.
If $|V(C_1)|=8$, then the vertices of the 4-diamond have cost $1.2$
and the vertices outside 4-diamonds have cost $1.3$ after phase 1.
The cost of the vertices outside 4-diamonds is decreased by at least $0.25$
with the first swap $C_1$ participates in in phase 2. Thus the the vertices of $C_1$
have cost less than $1.25$.
\end{proof}

\begin{lemma}\label{lemma6}
Let $G$ be a cubic bridgeless irreducible graph and let $F$ be a $2$-factor of $G$. If we do swapping operations on $F$ in the specified order, then after phase 2 for each 
$C\in {\cal C}_{\text{6-noint}}^*$, there are at least $4$ vertices
$v \in V(C)$ such that $c_2(v) \le 1.2$.
\end{lemma}
\begin{proof}
Three edges of $C$ are in $F$. If these three edges belong to three distinct circuits of 
$F$ then each of these circuits participates in a $6$-swap in phase 1 ($C$ by definition does not touch circuit from $C^*$), By Observation~\ref{swapcost}, the costs of the vertices of all three circuits is at most $1.2$ after phase 1 and thus also after phase 2.
If two edges of $C$ belong to the same circuit $C'$ of $F$, 
then $C'$ has at least $10$ vertices. 
The X-circuit $X$ of $F_1$ that includes $C'$ has at least $10$, $16$, $22$, and $28$ vertices if $C'$ is part of $0$, $1$, $2$, or $3$ swaps in phase 1.
By Observations~\ref{swapcost}--\ref{swapcost0} the vertices
of $C \cap C'$ (that is at least $4$ vertices) have cost at most 1.2.
\end{proof}

\section{Proof and algorithm}\label{sec6}

We are now almost ready to use Lemma~\ref{factorlemma} to prove Theorem~\ref{thm}.
In Section~\ref{sec5} we defined the set of good collections 
${\cal I}$ for a given fixed $2$-factor $F$.
We defined costs $c_2$ of vertices which will be used as the function $c$ of Lemma~\ref{factorlemma}.
We also bounded the values $c(v)$ and thus we will be able to
set the values of $s_{\cal H}$ and  $t_{\cal H}$ for each ${\cal H} \in {\cal I}$.
The last necessary ingredient is definition of non-intersecting sets 
${\cal \SV H}$ for each ${\cal H} \in {\cal I}$. As it turns out,
we can simply set ${\cal \SV H} = {\cal \V H}^*$ except when ${\cal H}={\cal C}_{\text{6-noint}}$.
We start by examining the intersections in the set
${\cal C}_{\text{6-noint}}^*$.

\begin{lemma}
Let $G$ be a cubic bridgeless graph and let $F$ be a $2$-factor of $G$.
There is a set ${\cal C}^I_{\text{6-noint}} \subseteq {\cal C}^*_{\text{6-noint}}$ such that circuits from ${\cal C}^I_{\text{6-noint}}$ are pairwise disjoint and 
$
|{\cal C}^I_{\text{6-noint}}|\ge \frac{1}{4} \cdot |{\cal C}^*_{\text{6-noint}}|.
$
\label{cki}
\end{lemma}
\begin{proof}
Let $C\in {\cal C}^*_{\text{6-noint}}$. All boundary edges of $C$ are in $F$ and
every second edge of $C$ is in $F$, while remaining edges are in the perfect matching $M$ complementary to $F$. 
First we show that there are at most 4 other circuits
$C'\in {\cal C}^*_{\text{6-noint}}$ intersecting $C$. 
Two intersecting circuits from ${\cal C}^*_{\text{6-noint}}$ must intersect on matching edges because every second edge in both circuits lies in the matching $M$. Let $e_1,e_2,e_3\in C \cap M$. Since the circuits $C$ and $C'$ are different, they cannot intersect on three matching edges, therefore $C$ and $C'$ can intersect in one or two edges from $\{e_1,e_2,e_3\}$. Let $C_i\in {\cal C}^*_{\text{6-noint}}$ be a circuit intersecting $C$ in exactly one matching edge $e_i$. If such a circuit exists, then it is unambiguously determined by $M$:
the edges adjacent to $e_i$ in $C_i$ must lie outside $C$ and the next two edges must be edges from $M$ which determines the $6$-circuit. 
Let $C_{i,j}\in {\cal C}^*_{\text{6-noint}}$ be a circuit intersecting $C$ in exactly two matching edges $e_i$ and $e_j$. 
If the circuit $C_{i,j}$ exists, then it is also unambiguously determined by $M$. 
Therefore, we can have at most $6$ disjoint circuits intersecting $C$: $C_1,C_2,C_3,C_{1,2},C_{2,3},C_{1,3}$. We show that only $4$ of them can exist at the same time. 
Let us assume that at least two of the circuits $C_{1,2}$, $C_{1,3}$, and $C_{2,3}$ exist.
Without loss of generality suppose that these two circuits are $C_{1,2}$ and $C_{1,3}$. 
Let us denote the vertices of $C$ so that $C=v_1v_2v_3v_4v_5v_6$, where 
$e_1=v_1v_6$, $e_2=v_2v_3$, and $e_3=v_4v_5$.
Let $w_1$, $w_2$, $w_3$, $w_4$, $w_5$, and $w_6$ be the neighbours of $v_1$, $v_2$, $v_3$, $v_4$, $v_5$, and $v_6$ outside $C$ (since $C$ has independent boundary they are pairwise distinct).
Then $C_{1,2}=v_1v_2v_3w_3w_6v_6$ and $C_{1,3}=v_1v_6v_5v_4w_4w_1$ (see Figure \ref{figC6}). 
\begin{figure}[htp]
\center
\includegraphics[scale=1.2]{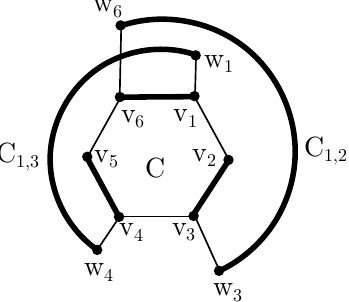}
\caption{Circuits $C_{1,2}$ and $C_{1,3}$.}
\label{figC6}
\end{figure} 
If the circuit $C_1$ exists, then it must be $v_1v_6w_6w_3w_4w_1$ and we have a $4$-circuit $v_3w_3w_4v_4$ intersecting $C$ which is a contradiction with $C\in {\cal C}^*_{\text{6-noint}}$.
Therefore, if both $C_{i,j}$ and $C_{i,k}$ exist, then the circuit $C_i$ does not. This shows that at most $4$ circuits from ${\cal C}^*_{\text{6-noint}}$ can intersect $C$ and at most two of the four intersecting
circuits intersect $C$ in two edges and the other two in one edge. 

Now suppose that there are five pairwise intersecting circuits in ${\cal C}^*_{\text{6-noint}}$.
Let $C$ be one of these circuits. Suppose first that exactly two of the circuits 
$C_{1,2}$, $C_{1,3}$, and $C_{2,3}$ exist, say $C_{1,2}$ and $C_{1,3}$, as
depicted in Figure \ref{figC6}. Then, by the above-mentioned argument, the remaining two existing circuits must be $C_2$ and $C_3$. 
Circuit $C_2$ must contain vertices $w_2$, $v_2$, $v_3$, $w_3$, $w_6$, 
but to have an intersection with $C_{1,3}$ it has to contain at least two vertices from $C_{1,3}$, a contradiction.
If only one of the circuits $C_{1,2}$, $C_{1,3}$, and $C_{2,3}$ exists, say $C_{1,2}$,
then it does not intersect $C_3$, a contradiction.
Therefore, there are no five pairwise intersecting circuits in ${\cal C}^*_{\text{6-noint}}$.

We create an auxiliary graph $H$ as follows. The vertices correspond to the circuits from 
${\cal C}^*_{\text{6-noint}}$ and two vertices will be joined by an edge if 
the corresponding circuits intersect. We know that all vertices of $H$ have degree at most four
and $H$ does not have $K_5$ as a subgraph.
Brooks theorem \cite[Theorem 5.2.4]{diestel} says that if 
a connected graph is not an odd cycle or a complete graph, then its chromatic number 
is less than equal to its maximum degree. Therefore,
there exists a set ${\cal C}^{I}_{\text{6-noint}}$ of non-intersecting circuits 
from ${\cal C}^*_{\text{6-noint}}$ such that
$
|{\cal C}^{I}_{\text{6-noint}}|\ge \frac{1}{4} |{\cal C}^*_{\text{6-noint}}|.
$
\end{proof}

\begin{lemma}
The subgraphs from
${\cal D}^*_4$,
${\cal D}^*_6$,
${\cal C}^*_{\text{4-noint}}$,
${\cal C}^*_{\text{5-noint}}$,
${\cal C}^I_{\text{6-noint}}$,
${\cal C}^*_{\text{4-4noint}}$, and
${\cal C}^*_{\text{4-int-5}}$ are mutually disjoint.
\label{lemmadisjoint}
\end{lemma}
\begin{proof}
Subgraphs from ${\cal D}^*_4$ are mutually disjoint as they are separated by a $2$-edge-cut. 
Moreover, any circuit with independent boundary that intersects a 4-diamond has length more than $6$, since it intersects the 4-diamond in four vertices, the two vertices neighbouring to a 4-diamond, and at least one additional vertex (otherwise the 4-diamond would be in a 6-diamond which by definition is not possible). 
Therefore, subgraphs from ${\cal D}^*_4$ are disjoint from all other subgraphs on the list.
Similar arguments  hold for ${\cal D}^*_6$.

By definition, circuits from ${\cal C}^I_{\text{6-noint}}$ do not intersects circuits from ${\cal C}^*$. Since ${\cal C}^4_{\text{4-noint}}$, ${\cal C}^4_{\text{5-noint}}$, ${\cal C}^4_{\text{4-4-noint}}$, and ${\cal C}^4_{\text{4-int-5}}$ are all subsets of ${\cal C}^*$, no circuits from ${\cal C}^I_{\text{6-noint}}$ intersect them.
Moreover, the circuits in this set are disjoint by definition.

All the remaining circuits belong to ${\cal C}^*$ and by definition do not touch any other circuit from ${\cal C}^*$, which implies that the induced subgraph must be disjoint
 which  concludes the proof of the lemma.
\end{proof}

Now we are ready prove the main result of this paper
by proving Theorem~\ref{thm}.
\begin{proof}[Proof of Theorem~\ref{thm}]
Let
$$
{\cal I}=\left\{ {\cal D}_4, {\cal D}_6, {\cal C}_{\text{4-noint}},
{\cal C}_{\text{5-noint}}, {\cal C}_{\text{4-4noint}}, 
{\cal C}_{\text{6-noint}}, {\cal C}_{\text{4-int-5}} \right\},
$$
let $r=1.3$, let the values $s_{\cal H}$, $t_{\cal H}$, and $p_{\cal H}/n_{\cal H}$ for ${\cal H} \in {\cal I}$
be as follows.
\begin{center}
\begin{tabular}{cccc}
${\cal H}$ & $s_{\cal H}$ & $t_{\cal H}$ & $p_{\cal H}/n_{\cal H}$ \\
\hline
${\cal D}_4$ & $1.5$ & $1.2$ & 1\\
${\cal D}_6$ & $4/3$ & $1.25$ & 1\\
${\cal C}_{\text{4-noint}}$ & $1.5$ & $1.2$ & 1\\
${\cal C}_{\text{5-noint}}$ & $1.4$ & $1.25$ & 1\\
${\cal C}_{\text{4-4-noint}}$ & $4/3$ & $1.2$ & 1\\
${\cal C}_{\text{6-noint}}$ & $4/3$ & $1.2$ & 1/6\\
${\cal C}_{\text{4-int-5}}$ & $1.375$ & $1.25$ & 1 
\end{tabular}
\end{center}
The equality (\ref{eqinlemma}) from the statement of Lemma~\ref{factorlemma}
holds for each ${\cal H} \in {\cal I}$. Let $F$ be the $2$-factor, whose existence is guaranteed by Lemma~\ref{factorlemma}. Let $F_2$ be the BE factor and 
for $v\in v(G)$ let $c(v)$ be the costs of vertices 
after phase 2 of performing swapping operations on $F$.
Let ${\cal \SV H}(F)={\cal \V H}^*(F)$, for 
${\cal H} \in {\cal I}- {\cal C}_{\text{6-noint}}$
and let ${\cal \SV C}_{\text{6-noint}}$ be the set of vertices 
that contains $4$ vertices of cost at most $1.2$ guaranteed by Lemma~\ref{lemma6}
for each of the $1/4 \cdot |{\cal C}^*_{\text{6-noint}}|$ disjoint $6$-circuits 
guaranteed by Lemma~\ref{cki}. The set ${\cal \SV C}_{\text{6-noint}}$ thus contains at least $|{\cal C}^*_{\text{6-noint}}|$ vertices.
Function $c$ satisfies the conditions of Lemma~\ref{factorlemma}
due to Lemmas \ref{lemma0}, \ref{lemmastar}, and \ref{lemmaall}.
By Lemma~\ref{factorlemma}, $F_2$ has cost at most $1.3\cdot |V(G)|$.
\end{proof}

Next we describe a polynomial time algorithm that given a bridgeless cubic graph $G$ on at least eight vertices
finds a TSP tour of length at most $1.3 \cdot |V(G)|-2$. We did not try to minimize the time complexity of the algorithm.
If graph contains a reducible subgraph, then we use one of the reductions
from Lemma~\ref{irred}. Each reducible subgraph can be found in linear time
and the size of $G$ is decreased by a constant factor. The reduction itself
can be done in constant time. Thus we can obtain an irreducible graph $G'$
using reductions from Lemma~\ref{irred} in quadratic time.
We find a BE factor of the irreducible graph $G'$ of cost at most $1.3 \cdot |V(G')|$
as follows. We can find all subgraphs in all sets of ${\cal I}$ in linear time as described in Lemmas~\ref{factorlemma}~and~\ref{circuit},
the only hard step is to find the convex combination of perfect matchings
that gives the point $(1/3, 1/3, \dots, 1/3)$ of the perfect matching polytope.
Barahona~\cite{Barahona} proved that such a combination that contains linear number of perfect matchings
can be found in $O(n^6)$ time. It is possible in linear time to find a circuit on which a swap can be performed, and as swaps merge circuits (Lemma~\ref{obsswaps}) only linearly many 
swaps can be done. Finally, we undo all reductions done to $G$. Each reduction can be reversed
in constant time. In this way we get an even factor $F$ of $G$ of cost at most $1.3 \cdot |E(G)|$.
Then we can find a spanning tree $T$ of $G/F$ in linear time.
Finally, we construct an Eulerian multigraph that contains edges from $T$ twice, edges of 
$F$ once, and we find an Eulerian tour in it, this is also possible in linear time. As each part of the algorithm runs in polynomial time, so does the whole algorithm.

\section*{Acknowledgements}
The work of the first author was partially supported by VEGA grant No. 1/1005/12.
The work of the second author was partially supported by VEGA grants No. 1/0042/14 and  1/0876/16.

\bibliographystyle{plain}
\bibliography{mybibl}

\begin{thebibliography}{10}

\bibitem{AS}
S.~Arora.
\newblock Polynomial time approximation schemes for euclidean traveling
  salesman and other geometric problems.
\newblock {\em J. ACM}, 45(5):753--782, September 1998.

\bibitem{Barahona}
F.~Barahona.
\newblock Fractional packing of t-joins.
\newblock {\em SIAM J. Discrete Math.}, 17:661--669, 2004.

\bibitem{BK}
P.~Berman and M.~Karpinski.
\newblock 8/7-approximationn algorithm for (1,2)-tsp.
\newblock In {\em Proceedings of the Seventeenth Annual ACM-SIAM Symposium on
  Discrete Algorithm}, SODA '06, pages 641--648, Philadelphia, PA, USA, 2006.
  Society for Industrial and Applied Mathematics.

\bibitem{BSSS}
S.~Boyd, R.~Sitters, S.~van~der Ster, and L.~Stougie.
\newblock The traveling salesman problem on cubic and subcubic graphs.
\newblock {\em Mathematical Programming}, 144(1-2):227--245, 2014.

\bibitem{C}
N.~Christofides.
\newblock Worst-case analysis of a new heuristic for the travelling salesman
  problem.
\newblock Technical Report 388, Graduate School of Industrial Administration,
  Carnegie Mellon University, 1976.

\bibitem{CLS}
J.~R. Correa, O.~Larré, and J.~A. Soto.
\newblock Tsp tours in cubic graphs: Beyond $4/3$.
\newblock {\em SIAM J. Discrete Math.}, 29:915--939, 2015.

\bibitem{diestel}
R.~Diestel.
\newblock {\em Graph Theory}.
\newblock Springer-Verlag New York, 5 edition, 2016.

\bibitem{DKM}
Z.~Dvorak, D.~Kral, and B.~Mohar.
\newblock Graphic tsp in cubic graphs.
\newblock arXiv:1608.07568 [cs.DM], 2016.

\bibitem{edmonds}
J.~Edmonds.
\newblock Maximum matching and a polyhedron with $0,1$ vertices.
\newblock {\em J. of Res. the Nat. Bureau of Standards}, 69~B:125--130, 1965.

\bibitem{KR}
J.~Karp and R.~Ravi.
\newblock A $9/7$-approximation algorithm for graphic tsp in cubic bipartite
  graphs.
\newblock arXiv:1311.3640 [cs.DS], 2013.

\bibitem{L}
Michael Lampis.
\newblock Improved inapproximability for tsp.
\newblock In Anupam Gupta, Klaus Jansen, José Rolim, and Rocco Servedio,
  editors, {\em Approximation, Randomization, and Combinatorial Optimization.
  Algorithms and Techniques}, volume 7408 of {\em Lecture Notes in Computer
  Science}, pages 243--253. Springer Berlin Heidelberg, 2012.

\bibitem{7-circuits}
R.~Lukot'ka.
\newblock Avoiding 7-circuits in 2-factors of cubic graphs.
\newblock {\em Electronic J Combinatorics}, 21(\#P4.21), 2014.

\bibitem{MS}
T.~Moemke and O.~Svensson.
\newblock Approximating graphic tsp by matchings.
\newblock In {\em Proceedings of the 2011 IEEE 52nd Annual Symposium on
  Foundations of Computer Science}, FOCS '11, pages 560--569, Washington, DC,
  USA, 2011. IEEE Computer Society.

\bibitem{zuylen}
A.~van Zuylen.
\newblock Improved approximations for cubic bipartite and cubic tsp.
\newblock {\em Lecture Notes in Computer Science}, 9682:250--261, 2016.

\bibitem{V}
Jens Vygen.
\newblock New approximation algorithms for the tsp.
\newblock {\em OPTIMA}, 90:1--12, 2012.

\end{thebibliography}

\end{document}